\documentclass[letterpaper, 10 pt, conference]{ieeeconf} 

\usepackage{amsmath,amssymb,amsfonts, bm}
\usepackage{mathtools}
\usepackage[hidelinks]{hyperref}
\usepackage{microtype}
\usepackage{xcolor}
\usepackage[normalem]{ulem}
\usepackage[caption=false,font=footnotesize]{subfig}
\usepackage{cite}
\usepackage[font=normalfont,font=footnotesize]{caption}

\newtheorem{theorem}{Theorem}[section]
\newtheorem{proposition}[theorem]{Proposition}
\newtheorem{lemma}[theorem]{Lemma}
\newtheorem{corollary}[theorem]{Corollary}
\newtheorem{remark}[theorem]{Remark}
\newtheorem{definition}[theorem]{Definition}
\newtheorem{example}[theorem]{Example}
\newtheorem{problem}{Problem}

\makeatletter
\renewcommand{\@begintheorem}[2]{\trivlist
  \item[\hskip \labelsep {\bfseries #1\ #2.}]}
\renewcommand{\@opargbegintheorem}[3]{\trivlist
  \item[\hskip \labelsep {\bfseries #1\ #2\ (#3).}]}
\makeatother

\renewcommand{\theenumi}{(\roman{enumi})}

\newcommand{\Bc}{\mathcal{B}}

\newcommand{\Oc}{\mathcal{O}}

\newcommand\cbf{\mathbf{c}}

\newcommand\gbf{\mathbf{g}}

\newcommand\ubf{\mathbf{u}}

\newcommand\xbf{\mathbf{x}}
\newcommand\ybf{\mathbf{y}}

\newcommand\Abf{\mathbf{A}}
\newcommand\Bbf{\mathbf{B}}
\newcommand\Cbf{\mathbf{C}}

\newcommand\Ibf{\mathbf{I}}

\newcommand\Pbf{\mathbf{P}}
\newcommand\Qbf{\mathbf{Q}}

\newcommand{\real}{{\mathbb{R}}}

\newcommand\zeros{{\mathbf{0}}}

\IEEEoverridecommandlockouts    
\allowdisplaybreaks

\newcommand{\longthmtitle}[1]{\mbox{}{\bf \textit{(#1).}}}

\newcommand{\oprocendsymbol}{\hbox{$\square$}}
\newcommand{\oprocend}{\relax\ifmmode\else\unskip\hfill\fi\oprocendsymbol}
\def\eqoprocend{\tag*{$\square$}}

\title{\LARGE \bf
Passivity-Based Control of Electrographic Seizures in a \\ Neural Mass Model of Epilepsy
}

\author{Gagan Acharya and Erfan Nozari%
\thanks{G. Acharya is with the Department of Electrical and Computer Engineering,
University of California, Riverside, CA 92521, USA
{\tt\small gacha002@ucr.edu}.}%
\thanks{E. Nozari is with the Departments of Mechanical Engineering, Electrical and Computer Engineering, Bioengineering, and the Neuroscience Graduate Program,
University of California, Riverside, CA 92521, USA
{\tt\small erfan.nozari@ucr.edu}.}%
}

\begin{document}

\maketitle

\begin{abstract}
Recent advances in neurotechnologies and decades of scientific and clinical research have made closed-loop electrical neuromodulation one of the most promising avenues for the treatment of drug-resistant epilepsy (DRE), a condition that affects over 15 million individuals globally. Yet, with the existing clinical state of the art, only 18\% of patients with DRE who undergo closed-loop neuromodulation become seizure-free. In a recent study, we demonstrated that a simple proportional feedback policy based on the framework of passivity-based control (PBC) can significantly outperform the clinical state of the art. However, this study was purely numerical and lacked rigorous mathematical analysis. The present study addresses this gap and provides the first rigorous analysis of PBC for the closed-loop control of epileptic seizures.
Using the celebrated Epileptor neural mass model of epilepsy, we analytically demonstrate that (i) seizure dynamics are, in their standard form, neither passive nor passivatable, (ii) epileptic dynamics, despite their lack of passivity, can be stabilized by sufficiently strong passive feedback, and (iii) seizure dynamics can be passivated via proper output redesign. To our knowledge, our results provide the first rigorous passivity-based analysis of epileptic seizure dynamics, as well as a theoretically-grounded framework for sensor placement and feedback design for a new form of closed-loop neuromodulation with the potential to transform seizure management in DRE.
% Passivity provides an energy-based framework for stability analysis and feedback design in nonlinear systems. This paper studies the passivity properties of the six-dimensional Epileptor neural mass model, a canonical dynamical model of focal seizure activity. Treating stimulation as an additive input, we analyze the control-affine structure of the model about the interictal equilibrium.

% We first show that the standard input–output pairing violates
% the matching condition of the Positive Real Lemma, and
% therefore the model is not locally passive under conventional
% sensing configurations. Nevertheless, certain passive feedback
% laws can still locally stabilize the system despite this structural
% obstruction. We then derive algebraic conditions under which
% output redesign restores passivity at the linearized level and
% provide explicit constraints on sensor placement that guarantee
% the existence of a quadratic storage function. Under these
% conditions, passive feedback ensures asymptotic stabilization
% of the equilibrium.

% The results provide the first systematic passivity analysis of the Epileptor model and establish structural conditions for passive neuromodulation design.

\end{abstract}

%%%%%%%%%%%%%%%%%%%%%%%%%%%%%%%%%%%%%%%%%%%%%%%%%%%%%%%%%%%%%%%%%%%%%%%%%%%%%%%%
\section{Introduction}\label{sec:intro}

%Epilepsy is a chronic neurological disorder that affects more than 50 million people worldwide and is characterized by the recurrent occurrence of seizures arising from abnormal, hypersynchronous neural activity in the brain~\cite{fisher2014ilae}. While anti-seizure medications provide effective control for many patients, approximately one third of individuals develop drug-resistant epilepsy (DRE), in which seizures persist despite adequate trials of pharmacological therapy. For these patients, surgical resection of the epileptogenic zone offers the highest chance of seizure freedom, but many patients are not candidates due to the location or distributed nature of their epileptic networks. As a result, implantable neuromodulation technologies have emerged as important therapeutic alternatives. In particular, the responsive neurostimulation (RNS) system detects abnormal electrographic activity and delivers electrical stimulation in closed loop to disrupt seizure evolution~\cite{jarosiewicz2021rns}. Despite promising progress, however, clinical outcomes remain highly variable, with only a minority of patients achieving seizure freedom~\cite{nair2020nine}. %These limitations highlight the need for new theoretical frameworks and control strategies that can more effectively regulate pathological neural dynamics underlying epileptic seizures.

Epilepsy is a chronic neurological disorder affecting more than 50 million people worldwide, characterized by recurrent seizures arising from abnormal, hypersynchronous neural activity~\cite{fisher2014ilae}. While anti-seizure medications effectively control seizures for many patients, approximately one-third develop drug-resistant epilepsy (DRE), in which seizures persist despite adequate pharmacological therapy. For these individuals, surgical resection of the epileptogenic zone offers the highest chance of seizure freedom, but many are not candidates due to the location or distributed nature of epileptic networks. Consequently, implantable neuromodulation technologies have emerged as important alternatives. In particular, the responsive neurostimulation (RNS) system detects abnormal electrographic activity and delivers electrical stimulation in a closed-loop manner to disrupt seizure evolution~\cite{jarosiewicz2021rns}. Despite this progress, clinical outcomes remain variable, with only a minority of patients achieving seizure freedom~\cite{nair2020nine}.

Motivated by the long-standing limitations of existing treatments, in a recent empirical study~\cite{acharya2026passive}, we proposed a radically different approach to closed-loop control of epileptic seizures based on the control-theoretic framework of passivity-based control (PBC). Passivity provides a robust, energy-based framework for analyzing the stability of nonlinear systems and feedback interconnections, 
% A passive system admits a storage function whose rate of change is bounded by the supplied power at an input–output port, and strict passivity directly yields Lyapunov stability under passive feedback. This structure makes 
making PBC particularly attractive for nonlinear control design~\cite{Khalil}. In the context of epilepsy, our proposed PBC-based approach uses neuromodulation to \textit{drain energy} from epileptic tissue, rather than injecting energy into it (status quo). Using extensive numerical analyses, we have shown the robust ability of this approach to suppress epileptic seizures and its significant advantage over the state of the art~\cite{acharya2026passive}. Nevertheless, this study was purely empirical and, to our knowledge, no theoretical analysis of passivity or PBC has been performed for any model of epilepsy. In the present work we address this gap and provide a rigorous analysis of the passivity and passivation of one of the most widely-used computational models of epilepsy.

% Seizure oscillations admit a natural energetic interpretation—persistent pathological activity may be viewed as sustained energy injection within the network—suggesting that dissipation-enforcing feedback could suppress such dynamics. Despite this structural compatibility, the passivity properties of commonly used seizure models have not been systematically characterized.

\textit{Related Work.}
In this paper we analyze the passivity properties of the Epileptor, one of the most widely used computational models of epileptic dynamics~\cite{Jirsa2014}. It provides a phenomenological neural-mass model of brain dynamics under epilepsy and has been widely used as a testbed for feedback control, including optimal feedback on networked variants \cite{Moosavi2023} and input shaping tailored to its dynamics \cite{Brogin2020}, and MPC-based neuromodulation evaluated on Epileptor-generated data \cite{Chatterjee2020}. Beyond methodological studies, the Epileptor has also served as a core dynamical model in several large-scale computational and translational projects. In particular, it has been integrated into whole-brain simulation frameworks which combine patient-specific structural connectomes with Epileptor-based network models to predict seizure propagation and support presurgical evaluation in clinical epilepsy cohorts \cite{jirsa2017virtual}.

To our knowledge, passivity-based feedback has not been studied for the Epileptor, or in the context of epilepsy in general. The clinically available neuromodulation system that is used for epilepsy~\cite{jarosiewicz2021rns} as well as various research studies that have sought to improve it~\cite{lundstrom2019chronic,zelmann2020closes,hamilton2018clinical,chadaide2025closed} implement \textit{active} stimulation paradigms that \textit{inject energy} into the tissue, often with empirically tuned stimulation parameters and no analytical guarantees. In modeling and control studies, seizure mitigation is commonly pursued via various active forms of proportional/integral feedback, or optimization-based designs such as LQR or MPC, with stability assessed through linearization, bifurcation analysis, or mere simulation~\cite{Schiff2010,Chatterjee2020,Brogin2020,Moosavi2023}. Related energy-oriented efforts have examined passivity observers \cite{Liu2019} and Hamiltonian formulations of neural dynamics \cite{vanderSchaftJeltsema2014}. However, a structural passivity characterization of the Epileptor %—particularly explicit input–output matching and sensor-placement conditions for quadratic storage functions—
has been lacking. The present study addresses this gap and provides the first passivity-based analysis of seizure dynamics.

\textit{Statement of Contributions.}
Our main contributions are threefold. 
First, we prove the asymptotic stability of the Epileptor under passive feedback, characterize the controller parameters that achieve this, and provide a sparse interpretable Lyapunov function as well as an estimate of the region of attraction of the asymptotically stable equilibrium point for a nominal controller.
Second, and despite the aforementioned stabilization, we prove that the Epileptor model is not passive in its standard form, both due to a lack of internal energy dissipation and a lack of proper input-output matching. We further establish an explicit necessary constraint on the Epileptor's output (i.e., sensor placement condition) for it to be able to be passive with a smooth storage function.
Finally, motivated by the latter output constraint, we provide an optimization-based framework for feedback passivation of the Epileptor, prove the strict passivity of the closed-loop system under a nominal set of controller parameters, and demonstrate the presence of a general tradeoff between asymptotic stability and passivity as well as a sparsity-promoting cost function for addressing it.
Together, these results provide the first theoretical analysis of PBC for closed-loop seizure suppression, and a rigorous foundation for our recent numerical observations of the great promise of PBC in epilepsy~\cite{acharya2026passive}.
% We first show that the standard input–output pairing violates the matching condition of the Positive Real Lemma, precluding local passivity under conventional sensing configurations. We then derive algebraic output conditions that restore passivity of the linearized system and establish explicit sensor-placement constraints guaranteeing the existence of a quadratic storage function. These results reveal a direct connection between output design and energy dissipation. Under the resulting matching conditions, passive feedback stabilizes the equilibrium. 
\section{Preliminaries}

In this section, we will review some preliminaries that form the foundation for the forthcoming analyses. 

% \subsection{Notation}

% For a square matrix $\Abf$, the \emph{spectral abscissa} is defined as
% \[
% \alpha(\Abf) := \max_{\lambda \in \sigma(\Abf)} \operatorname{Re}(\lambda),
% \]
% where $\sigma(\Abf)$ denotes the set of eigenvalues of $\Abf$.

\subsection{The Epileptor}\label{sec:prelim_ep}

The Epileptor is perhaps the most widely-used mean-field model of seizure dynamics~\cite{Saggio2020}. %It consists of a six-dimensional nonlinear state-space system of ordinary differential equations that reproduces several empirically observed aspects of focal seizures through coupled slow–fast subsystems~\cite{Jirsa2014}. 
Its structure reproduces the
coexistence of a stable equilibrium (interictal state) and a
stable limit cycle (ictal oscillations) via slow modulation of
a permittivity variable. This intrinsic bistability makes the
model particularly suitable for stability and passivity analysis
from a nonlinear control perspective.

The Epileptor dynamics consist of a system of $n = 6$ ordinary differential equations, given by 
\begin{subequations}\label{eq:dyn}
\begin{align}
\dot{x}_1 &= y_1 - f_1(x_1, x_2, z) - z + I_{rest,1} + u 
\\
\tau_1 \dot{y}_1 &= y_0 - 5 x_1^2 - y_1 
\\
\dot{x}_2 &= -y_2 \!+\! x_2 \!-\! x_2^3 \!+\! 2\zeta \!-\! 0.3(z\!-\!3.5) + I_{rest,2} + u 
\\
\tau_2 \dot{y}_2 &= -y_2 + f_2(x_2) 
\\
\dot{\zeta} &= -\gamma (\zeta - 0.1 x_1) 
\\
\tau_0 \dot{z} &= 4(x_1 - x_0) - z
\end{align}
\end{subequations}
where 
\begin{align*}
f_1(x_1,x_2,z) &=
\begin{cases}
x_1^3-3x_1^2, & x_1<0 \\
\big(x_2-0.6(z-4)^2\big)x_1, & x_1\ge 0
\end{cases} ,
\\
f_2(x_2) &=
\begin{cases}
6(x_2+0.25), \quad\qquad\ \  & x_2\ge -0.25 \\
0, & x_2<-0.25
\end{cases},
\end{align*}
and the constant parameters $x_0, y_0, \tau_0, \tau_1, \tau_2, I_{rest,1}, I_{rest,2}, \gamma$ are given in Table~\ref{tab:epileptor_params}. In particular, the constant currents $I_{rest,1}$ and $I_{rest,2}$
represent tonic baseline inputs to the fast and slow subsystems, respectively, and model persistent background excitation. The exogenous control input $u$ adds to these background inputs and will be used in the sequel for feedback control. % and set the operating regime of the autonomous system by shifting equilibria and influencing the presence of bistability between interictal and ictal attractors.
In this model, the linear combination of the fast and slow population activities, i.e.,
\begin{align}\label{eq:y}
y = c_1 x_1 + c_3 x_2 %= c_{\mathrm{std}}\, \xbf(t),
\end{align}
%with $c_{\mathrm{std}}=\begin{bmatrix}1 & 0 & 0 & -1 & 0 & 0\end{bmatrix}$
has been found to exhibit similar dynamics to empirically recorded local field potentials (LFPs) in patients with epilepsy. $y$ is thus often taken as the output, with the default values
\begin{align}\label{eq:stdc1c2}
    c_1 = -c_3 = 1.
\end{align}

\begin{table}[b]
\centering
\small
\caption{The default values of the Epileptor parameters~\cite{Jirsa2014}. The same values are used throughout this work.}
\label{tab:epileptor_params}
\begin{tabular}{p{0.6\linewidth} p{0.1\linewidth}}
\hline
\textbf{Parameter} & \textbf{Value} \\
\hline
Fast subsystem offset, $x_0$ & $-1.6$ \\
Fast subsystem baseline, $y_0$ & $1$ \\
Fast time constant, $\tau_1$ & $1$ \\
Slow time constant, $\tau_0$ & $2857$ \\
Slow subsystem time constant, $\tau_2$ & $10$ \\
Resting input, $I_{\mathrm{rest},1}$ & $3.1$ \\
Resting input, $I_{\mathrm{rest},2}$ & $0.45$ \\
Permittivity damping, $\gamma$ & $0.01$ \\
\hline
\end{tabular}
\end{table}

Given the separation of timescales imposed by the values of $\tau_0, \tau_1$ and $\tau_2$, the states
\begin{align*}
\xbf = 
\begin{bmatrix}
x_1 & y_1 & x_2 & y_2 & \zeta & z
\end{bmatrix}^\top \in \mathbb{R}^n,
\end{align*}
of the Epileptor are often grouped into the `fast subsystem' $(x_1,y_1)$ generating
rapid discharges, the `slow subsystem' $(x_2,y_2)$ exhibiting spike–wave discharges, the `ultraslow subsystem' $z$ governing the transitions between dynamical
regimes (i.e., into and out of seizures, cf. Fig.~\ref{fig:sim_baseline}), and the adaptation variable $\zeta$ that couples the fast and slow subsystems. For nominal parameter values (Table~\ref{tab:epileptor_params}), the autonomous
dynamics ($u = 0$) has no
asymptotically stable equilibria and instead exhibits a stable limit cycle, where in each period the output ($y$) transitions between high-frequency ictal (seizure)
oscillations and a meta-stable interictal (normal) baseline (see Fig.~\ref{fig:sim_baseline}). The latter transitions, importantly, are driven by the ultraslow state $z$ that acts as a bifurcation parameter for the remaining states. For more detailed discussions of the autonomous Epileptor dynamics, see~\cite{Jirsa2014}.

\begin{figure}
\centering
\subfloat[]{
\includegraphics[width=\linewidth]{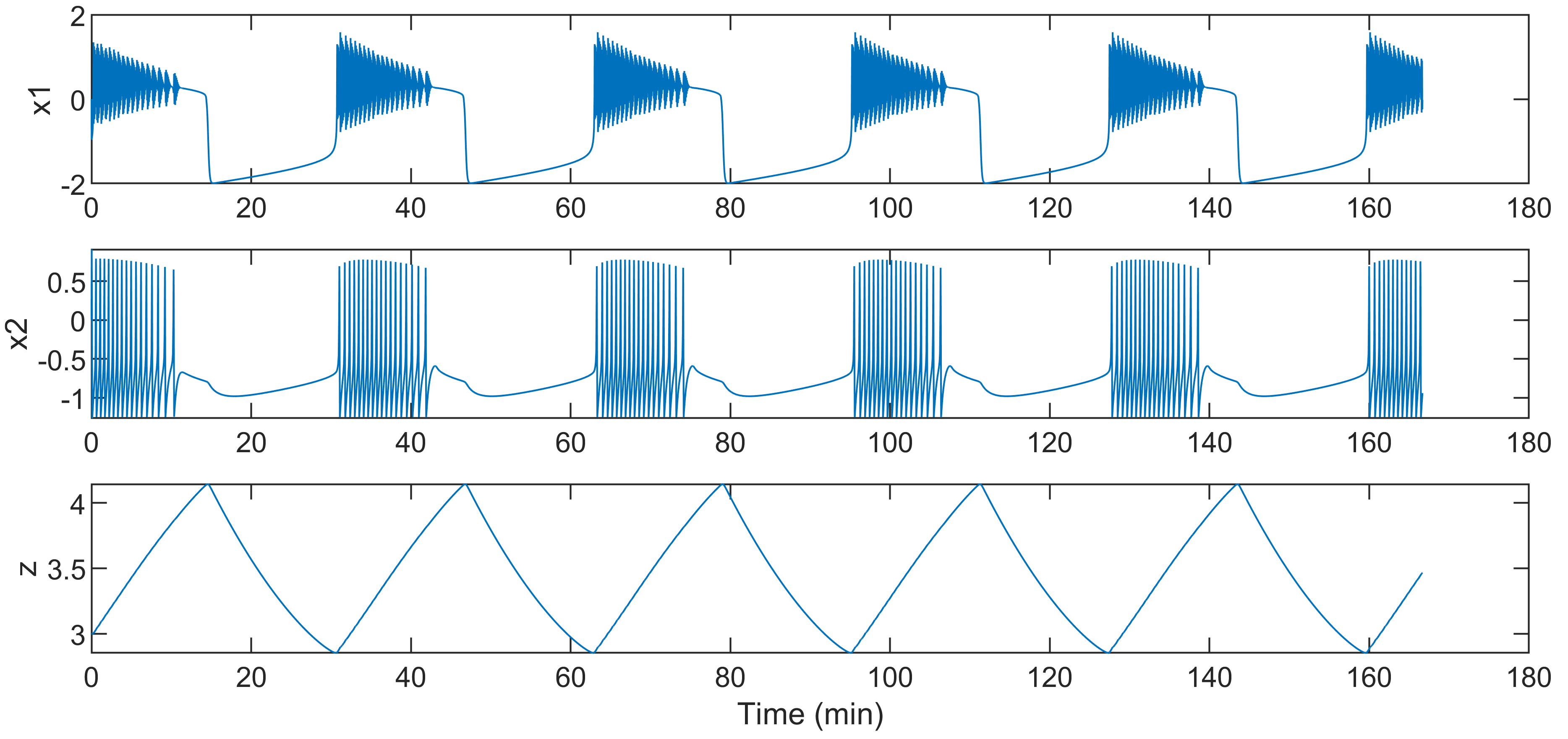}
}
\\[-1pt]
\subfloat[]{
\includegraphics[width=\linewidth]{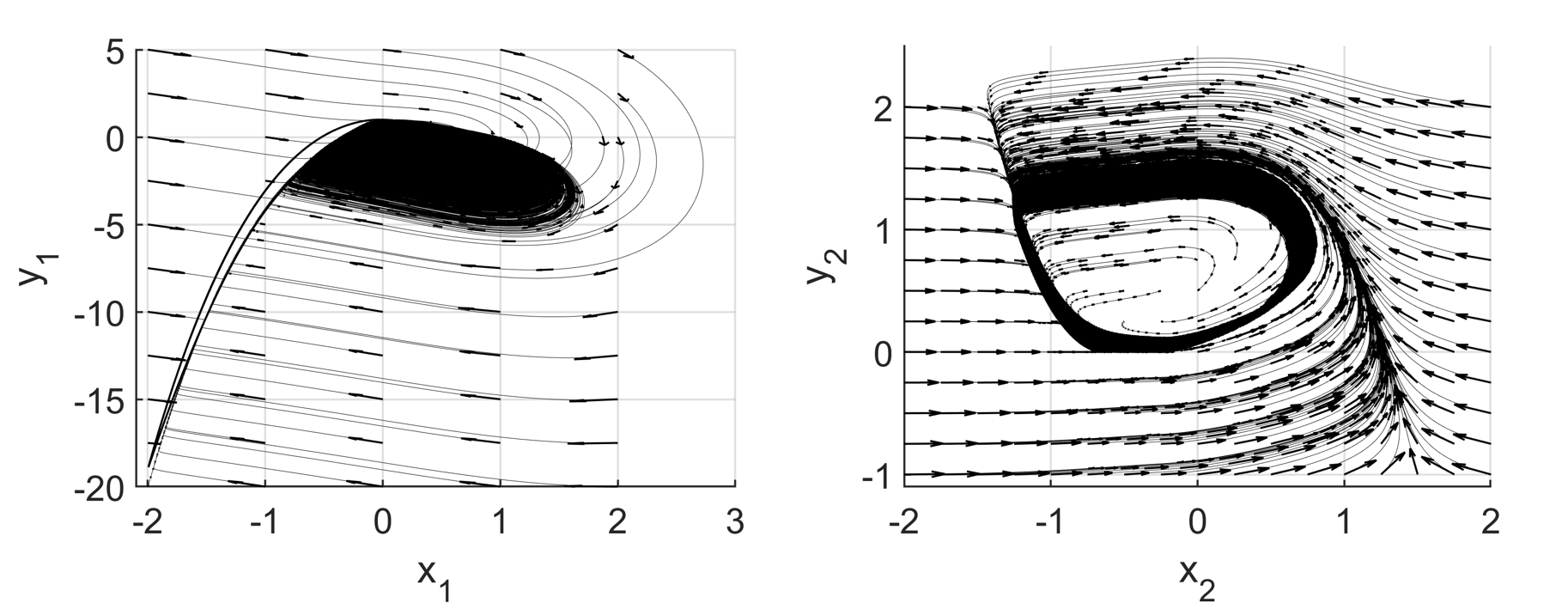}
}
\caption{(a) Simulated state trajectories of the autonomous Epileptor model ($u = 0$). One state is selected from each subsystem for illustration.
% Top: $x_1$. Middle: $x_2$. Bottom: slow permittivity variable $z$. 
The ultraslow modulation by $z$ drives bifurcations in the remaining (faster) states between interictal and ictal regimes.
(b) Phase–plane trajectories of the autonomous Epileptor model (same model as in (a)) plotted separately for the fast $(x_1,y_1)$ and slow $(x_2,y_2)$ subsystems.
% Left: $(x_1,y_1)$ subsystem. Right: $(x_2,y_2)$ subsystem. 
Trajectories converge to a stable limit cycle corresponding to ictal oscillations.}
\label{fig:sim_baseline}
\end{figure}

% \begin{figure}
%     \centering
    
%     \caption{}
%     \label{fig:flow_baseline}
% \end{figure}

Note that~\eqref{eq:dyn} has the control-affine form
% In the Epileptor equations, the baseline currents
% $I_{rest,1}$ and $I_{rest,2}$ appear additively in the
% $\dot{x}_1$ and $\dot{x}_2$ dynamics. Motivated by this
% structure, we model external stimulation as an additional
% current entering through the same channels. Introducing a
% scalar input $u(t)$, the system can therefore be written in control-affine form
\begin{align}\label{eq:affine}
\dot{\xbf} = f(\xbf) + \gbf u
\end{align}
where $f(\xbf)$ denotes the autonomous vector field
%(including all nonlinear terms and the constants
% $I_{rest,1}, I_{rest,2}$), 
and
\begin{equation}\label{eq:g}
\gbf =
\begin{bmatrix}
1 & 0 & 1 & 0 & 0 & 0
\end{bmatrix}^\top.
\end{equation}
% The vector $g$ reflects that stimulation perturbs only themembrane potential variables $x_1$ and $x_2$, while the remaining states evolve indirectly through intrinsic nonlinear couplings. 
This explicit input structure will play a central role
in the subsequent passivity analysis.

\subsection{Passivity}\label{subsec:prelim-passivity}

In this section, we recall several standard results on passivity for nonlinear
control-affine systems of the form
\begin{subequations}\label{eq:affine_recalled}
\begin{align}
\dot{\xbf} &= f(\xbf) + \gbf u,
\\ 
y &= h(\xbf),
\end{align}
\end{subequations}
where $\xbf \in \mathbb{R}^n$ and $u,y \in \mathbb{R}$, with $f(\mathbf{0})=\mathbf{0}$ and $h(\mathbf{0})=0$.
The notion of passivity formalizes physical systems that dissipate energy through the
existence of a storage function whose derivative bounds the energy supplied to the system through its input–output port, as given next.

\begin{definition}\longthmtitle{Passivity~\cite{Khalil}}
System \eqref{eq:affine_recalled} is \emph{passive} if there
exists a continuously differentiable positive semidefinite storage function
$V:\mathbb{R}^n \to \mathbb{R}_{\ge 0}$ such that
\begin{equation}
\dot{V}(\xbf) \le u\,y,
\label{eq:passivity_def}
\end{equation}
for all $(\xbf,u) \in \real^{n+1}$. The system is called \emph{strictly passive} if 
\[
\dot{V}(\xbf) \le u\,y - \kappa(\xbf)
\]
for all $(\xbf,u) \in \real^{n+1}$ and some positive definite function $\kappa$. The system is called \emph{locally passive} about the equilibrium $(\xbf^\star,u^\star)=(\mathbf{0},0)$ if~\eqref{eq:passivity_def} holds for all $(\xbf,u)$ in some neighborhood $\mathcal{N}$ of $(\xbf^\star,u^\star)$. 
\oprocend 
\end{definition}

For control-affine systems of the form \eqref{eq:affine_recalled},
the passivity inequality can be expressed directly in terms
of the system vector fields. Since
\begin{equation*}
\dot{V}(\xbf)
= \nabla V(\xbf)^\top f(\xbf)
+ \nabla V(\xbf)^\top \gbf\,u,
\end{equation*}
the condition in \eqref{eq:passivity_def} becomes
\[
\nabla V(\xbf)^\top f(\xbf)
+ \big(\nabla V(\xbf)^\top \gbf - h(\xbf)^\top\big)u
\le 0.
\]
Since this inequality must hold for all inputs
$u \in \real$, it follows that the so-called \textit{matching condition}
\begin{equation}\label{eq:matching}
\nabla V(\xbf)^\top \gbf = h(\xbf)^\top
\end{equation}
is necessary for passivity of~\eqref{eq:affine_recalled}. Under \eqref{eq:matching}, passivity reduces to the internal dissipation condition
\begin{equation}\label{eq:dissipation}
\nabla V(\xbf)^\top f(\xbf) \le 0.
\end{equation}

The passivity of a dynamical system has important implications for its feedback
stabilization. In particular, strictly passive systems
can be stabilized using simple output feedback laws, as shown next.

\begin{proposition}\longthmtitle{Feedback stabilization of strictly passive systems~\cite{Khalil}}\label{prop:asym-stable}
Assume that system \eqref{eq:affine_recalled} is \emph{strictly passive} with storage function $V(\xbf)$ and $\phi:\mathbb{R}\to\mathbb{R}$ is a (passive) function that satisfies 
\begin{align}\label{eq:yphiy}
y\,\phi(y) > 0 \quad \text{for all } y \neq 0.
\end{align}
Then, the static output feedback law
\begin{equation}\label{eq:passive_feedback}
u = -\phi(y),
\end{equation}
makes the origin of the closed-loop system asymptotically stable since
\[
\dot V(\xbf)
\le - y\,\phi(y) - \kappa(\xbf)
< 0 \quad \text{for} \quad \xbf \neq 0. \eqoprocend
\]
\end{proposition}

Similar to standard stability analysis, passivity can be generalized around any equilibrium point using a change of coordinates.
% In many applications, stability must be established not
% around the origin but around a general equilibrium point.
% Passivity properties can therefore be formulated relative
% to such operating points through an appropriate coordinate shift.
% \textbf{Strict passivity about an equilibrium.}
In particular, let $(\xbf^\star,u^\star)$ be a forced equilibrium of~\eqref{eq:affine_recalled} satisfying
\begin{align}\label{eq:equilibrium}
0 = f(\xbf^\star) + \gbf u^\star,  
\qquad
y^\star = h(\xbf^\star).
\end{align}
Passivity inequalities then need to be satisfied by the shifted variables
\begin{align}\label{eq:tildes}
\tilde \xbf = \xbf - \xbf^\star,
\qquad
\tilde u = u - u^\star,
\qquad
\tilde y = y - y^\star,
\end{align}
% The system is \emph{strictly passive about} $(\xbf^\star,u^\star)$
% if there exists a storage function $V(\tilde x)$ such that
% \[
% \dot V \le \tilde u\,\tilde y - \alpha(\tilde x),
% \]
% for some positive definite function $\alpha$. Under the feedback
and the passive feedback law
\begin{equation}
u = u^\star - \phi(y-y^\star),
\end{equation}
% with $\eta\,\phi(\eta)>0$ for $\eta\neq 0$,
% an argument identical to the non-shifted case yields
% \[
% \dot V < 0 \quad \text{for } \xbf \neq \xbf^\star.
% \]
% Therefore, $\xbf^\star$ is asymptotically stable.
makes a strictly passive system asymptotically stable around $\xbf^\star$ by Proposition~\ref{prop:asym-stable}.

While the above notions specialize naturally to linear systems, the passivity of a linear system admits a complete algebraic characterization, as follows.

\begin{proposition}\longthmtitle{Linear passivity conditions~\cite{Khalil,AndersonVongpanitlerd}}\label{prop:prl}
Consider the linear system
\begin{align*}
\dot \xbf &= \Abf \xbf + \Bbf \ubf, \\
\ybf &= \Cbf \xbf .
\end{align*}
Then the system is passive with a quadratic storage function
\[
V(\xbf)=\frac{1}{2}\xbf^\top \Pbf \xbf
\]
if and only if there exists a symmetric matrix $\Pbf=\Pbf^\top$ satisfying
\begin{subequations}\label{eq:PRL}
\begin{align}
\Pbf &\succeq \zeros, \label{eq:PRLa}\\
\Abf^\top \Pbf + \Pbf \Abf &\preceq \zeros, \label{eq:PRLb}\\
\Pbf \Bbf &= \Cbf^\top. \label{eq:PRLc}
\end{align}
\end{subequations}
If $\Abf^\top \Pbf + \Pbf \Abf \prec \zeros$, the system is strictly passive.
\oprocend
\end{proposition}
\begin{proof}
The sufficiency is shown in~\cite[Lemma 6.4]{Khalil}. The necessity follows from the quadratic form of $V$ and the fact that the dissipative inequality~\eqref{eq:passivity_def} yields
\begin{align*}
\dot V(\xbf) -\ubf^\top\ybf= -\frac{1}{2}\xbf^\top\Qbf\xbf+\xbf^\top(\Pbf\Bbf-\Cbf^\top)\ubf\leq0,
\end{align*}

where $\Qbf=-(\Abf^\top \Pbf + \Pbf \Abf)$.
\end{proof}

Note that in Proposition~\ref{prop:prl}, condition~\eqref{eq:PRLc} enforces an input–output
matching relation while~\eqref{eq:PRLb} guarantees
internal stability. Finally, similar to asymptotic stability, the passivity of a local linearization of a nonlinear system is related to the local passivity of the latter, as shown next.

%\erfan{(1) The following seem if and only if. What's the difference? (2) I couldn't find the latter one in Khalil, can you cite specific results that you use (also for other results in this section)? (3) I don't have Sepuchre's book, let's check its exact formulation together.}

% \begin{lemma}[Local Necessity]
% Let $\xbf^\star$ be an equilibrium of \eqref{eq:affine_recalled}. 
% If the nonlinear system is locally passive about $\xbf^\star$ with twice continuously differentiable storage $V$ satisfying $\nabla^2 V(\xbf^\star)=P\succ0$, then the linearization $(A,B,C)$ at $\xbf^\star$ satisfies \eqref{eq:PRL}.
% \hfill\cite{Sepulchre}
% \end{lemma}

\begin{proposition}\longthmtitle{Local strict passivity from linearization~\cite[Cor. 4]{Xia2015}}
\label{prop:local-suff}
Consider the nonlinear system in~\eqref{eq:affine_recalled} and assume that its linearization around an equilibrium point $(\xbf^\star, u^\star, y^\star)$ is strictly passive. Then the nonlinear system is locally strictly passive about $(\xbf^\star, u^\star, y^\star)$. \oprocend
% In particular, there exists a neighborhood $\mathcal{N}$ of $(\xbf^\star,u^\star)$ and a positive definite storage function $V(\xbf)$ such that for all $(\xbf,u)\in\mathcal{N}$,
% \begin{align*}
% (u-u^\star) (y-y^\star) - \dot V(\xbf) \ge 0 . \eqoprocend
% \end{align*}
\end{proposition}

% \begin{proof}
% This result follows from Theorem~2 and Corollary~4 in~\cite{Xia2015}.
% \end{proof}

\section{Problem Formulation}

As noted in Section~\ref{sec:intro}, this work is motivated by our recent empirical observations~\cite{acharya2026passive} as well as the structural robustness of PBC to model uncertainty, a property that is invaluable in feedback control of a complex system such as the brain. Our main goal is to understand what \textit{guarantees} we can provide for the empirically-observed stability of the Epileptor model under passive feedback. We approach this goal in two steps, as formulated next.

\begin{problem}\label{prob}
Consider the Epileptor model in~\eqref{eq:dyn}-\eqref{eq:stdc1c2} and assume that only the output $y$ is available for feedback. Determine whether, and under what conditions,
\begin{itemize}
    \item[(a)] epileptic seizures in the model can be \textit{stabilized} using passive feedback;
    \item[(b)] the system can be \textit{passivated} via feedback. \oprocend
\end{itemize}
\end{problem}

\begin{remark}
Note that the objective of Problem~\ref{prob}b is stronger than that of Problem~\ref{prob}a. Given the input-affine nature of the model, a standard feedback form for PBC is (\cite[Eq.~(14.74)]{Khalil})
\begin{align}\label{eq:fbform}
    u = u^\star - \phi(y - y^\star) + v
\end{align}
where $u^\star$ and $y^\star$ represent the target equilibrium (as in~\eqref{eq:equilibrium}), $\phi(\cdot)$ denotes a passive output feedback law as in~\eqref{eq:yphiy}, and $v$ provides additional capacity for control. In this context, Problem~\ref{prob}a only seeks asymptotic stability when $v = 0$, whereas Problem~\ref{prob}b seeks passivity with respect to $v$. The latter then guarantees that the system remains closed-loop stable under any passive feedback from $y$ to $v$, creating an opportunity for adding arbitrarily-many additional elements to the overall seizure suppression system (e.g., seizure detection, saturations for safety limits, low/band/high-pass filters, etc.) as long as the added elements remain passive (e.g., use passive electronic elements). \oprocend
\end{remark}

We address Problem~\ref{prob}a in Section~\ref{sec:passivity_standard} and Problem~\ref{prob}b in Section~\ref{sec:motiv_output}

%%%%%%%%%%%%%%%%%%%%%%%%%%%%%%%%%%%%%%%%%%%%%%%%%%%%%%%%%%%%%%%%%%%%%%%%%%%%%%%%
\section{Stabilization of the Epileptor via Passive Feedback}\label{sec:passivity_standard}

In this section, we focus on the ability of passive feedback to stabilize the Epileptor model, as motivated by our recent affirmative empirical observations~\cite{acharya2026passive} and formulated in Problem~\ref{prob}. 
To begin, let 
\begin{subequations}\label{eq:lin}
\begin{align}
\dot{\tilde{\xbf}} &= \Abf \tilde{\xbf} + \gbf \tilde{u}, 
\\ 
\tilde y &= \cbf^\top \tilde{\xbf},
\end{align}
\end{subequations}
denote the linearization of~\eqref{eq:affine},~\eqref{eq:y} around an equilibrium point ($x^\star$, $u^\star$), with
\begin{align}\label{eq:stdc}
\cbf = \begin{bmatrix}
    c_1 & 0 & c_3 & 0 & 0 & 0
\end{bmatrix}^\top
\end{align}
and $\tilde{\xbf}$, $\tilde{u}$, and $\tilde{y}$ denoting the centered variables in~\eqref{eq:tildes}. 
%Throughout this section, we consider scalar outputs $y=c\xbf$ with $c \in \mathbb{R}^{1\times 6}$.
%
As summarized in Section~\ref{subsec:prelim-passivity}, if a system is \textit{strictly passive}, it can be asymptotically stabilized by any form of passive feedback. The first question that arises, therefore, is whether (locally) strictly passive and, perhaps, that allows for its asymptotic stabilization under passive feedback.
The next result formalizes the rather straightforward observation that this is \textit{not} the case.

\begin{lemma}\longthmtitle{Non-passivity of the Epileptor}\label{lem:nonpassive}
The standard Epileptor model~\eqref{eq:dyn}-\eqref{eq:stdc1c2} is not strictly passive around any equilibrium points.
\end{lemma}
\begin{proof}
The unforced model ($u = 0$) has 4 equilibrium points at
\begin{align*}
    \xbf^\star_1 &= \begin{bmatrix}
-0.75 &  -1.82  &  -0.74  & 0 & -0.075 &   3.39
\end{bmatrix}^\top, \\ 
\xbf^\star_2 &= \begin{bmatrix}
-0.75 &  -1.82 &  -0.23  & 0.11 & -0.075  &   3.39
\end{bmatrix}^\top, \\ 
\xbf^\star_3 &= \begin{bmatrix}
-0.75 &  -1.82  &  -0.39  & 0 & -0.075 &   3.39
\end{bmatrix}^\top, \\ 
\xbf^\star_4 &= \begin{bmatrix}
0.43 &  0.07  &  -1.28  & 0 & 0.043 &   8.12
\end{bmatrix}^\top,
\end{align*}
The linearization~\eqref{eq:lin} is nevertheless unstable at all of these equilibrium points, with 
\begin{align*}
\alpha(\Abf_1) = 0.1766,  \quad \alpha(\Abf_2) = 0.3698, \\ \alpha(\Abf_3) = 0.5416, \quad \alpha(\Abf_4) = 11.136, 
\end{align*}
where $\alpha(\cdot)$ denotes spectral abscissa and each $\Abf_i = \frac{\partial f}{\partial \xbf} (\xbf^\star_i)$. Therefore, the unforced nonlinear model is also unstable around all its equilibria~\cite [Thm 4.7]{Khalil} and, hence, cannot be strictly passive~\cite[Lem 6.7]{Khalil}.
\end{proof}

% In this section we start the analysis of the passivity of the Epileptor model with the standard choice of output in~\eqref{eq:y},~\eqref{eq:stdc1c2}. 

Despite Lemma~\ref{lem:nonpassive}, strict passivity is only a sufficient condition, and stabilization of Epileptor seizures is still possible via passive feedback. We thus next seek to characterize, both numerically and analytically, the conditions under which such stabilization can be achieved (Problem~\ref{prob}a).

% \subsection{Passive-shunt feedback under the standard output}

% Corollary~\ref{cor:std_not_passive} established that the
% standard LFP proxy output $y = x_1 - x_2$ violates
% the passivity matching condition and therefore the
% Epileptor linearization is not passive under the
% conventional input–output pairing. 
% Nevertheless, passivity is only a sufficient condition
% for stabilization. We now show that a passive-type
% feedback law can still locally stabilize the system.

% We continue to use the standard LFP proxy output
% \[
% y = x_1 - x_2 = c_{\mathrm{std}} \xbf,
% \qquad
% c_{\mathrm{std}} = \begin{bmatrix}1&0&0&-1&0&0\end{bmatrix}.
% \]
% We consider a
% \emph{passive-shunt} feedback law of the form used in our prior work,
% which implements dissipative output feedback without requiring the plant
% passivity.

Consider again the general feedback form~\eqref{eq:fbform} with $v = 0$ (cf. Section~\ref{sec:motiv_output} for nonzero $v$) and let, for simplicity,
\begin{align}\label{eq:linear-phi}
\phi(y) = k y, \qquad k \ge 0.
\end{align}
The input,
\begin{align}\label{eq:shunt_feedback}
u = u^\star - k(y - y^\star)
\end{align}
is then fully determined by the constant baseline input $u^\star$ determining system equilibria and feedback gain $k$. To understand the ability of this input to stabilize the system, we first numerically swept this two-dimensional parameter space and tracked the stability of the linearization in~\eqref{eq:lin}. 

The resulting spectral abscissa of closed-loop Jacobian matrix $\Abf - k \gbf \cbf^\top$ is shown in Fig.~\ref{fig:heatmaps}. Since the system is not strictly passive, there exist passive controllers that fail to stabilize it (the white region in the figure). Yet, `sufficiently strong' passive controllers can asymptotically stabilize the system, namely, when $u^\star$ is sufficiently negative and/or $k$ is sufficiently large.

\begin{figure}[t]
\centering
% \subfloat[Spectral Abscissa]{%
\includegraphics[width=0.5\linewidth]{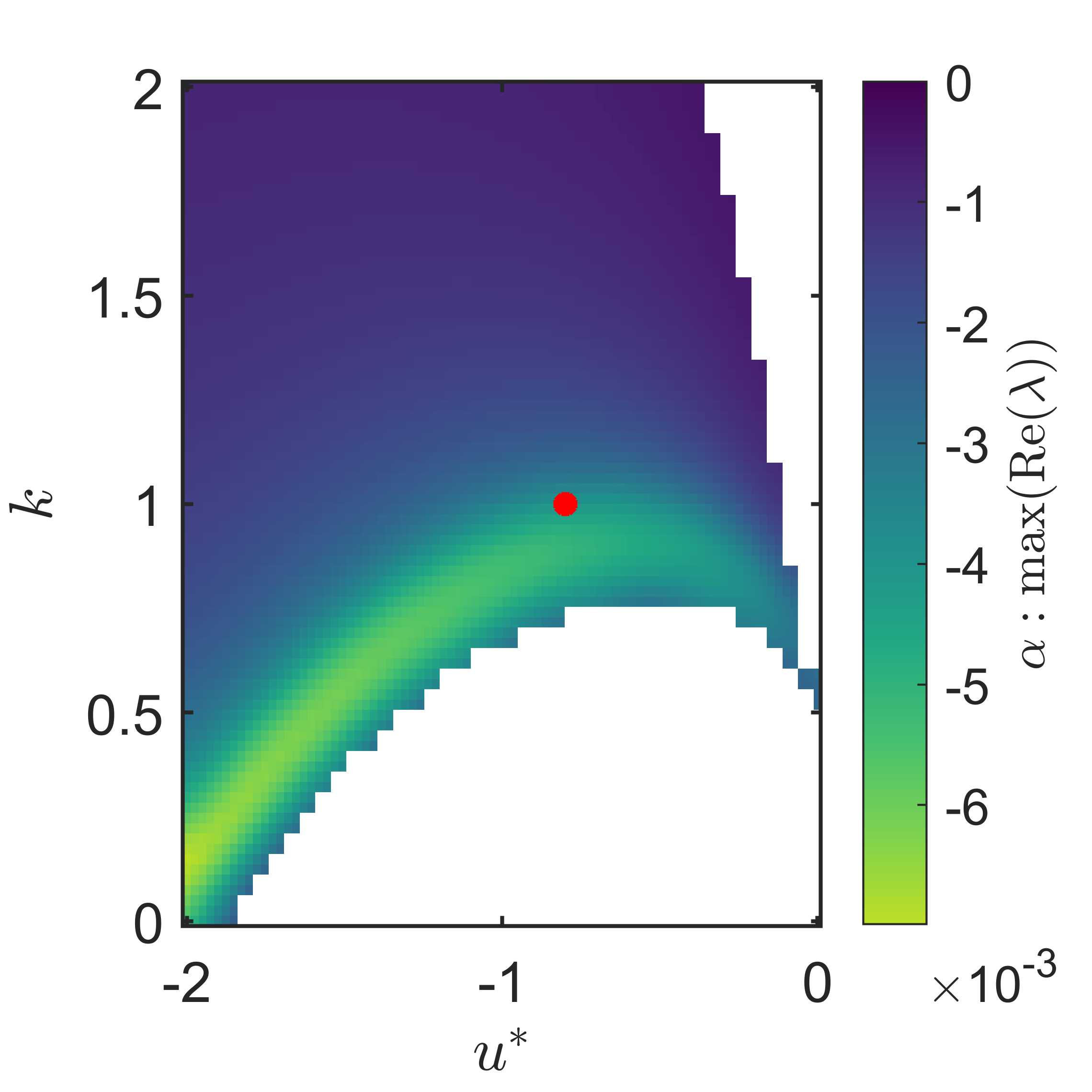}
% }
% \hfill
% \subfloat[Steady State]{%
% \includegraphics[width=0.48\linewidth]{figures/cmap_steady_1_minus_1.png}}
    \caption{The spectral abscissa of the closed-loop Jacobian matrix $\Abf - k \gbf \cbf^\top$ (cf.~\eqref{eq:lin}) as a function of the control parameters $(u^\star, k)$ of the passive feedback controller in~\eqref{eq:shunt_feedback}. Control parameters that fail to stabilize the system ($\alpha(\Abf) > 0$ are shown as white. The red dot corresponds to the nominal parameters $u^\star = -0.8, k = 1$ used in~\cite{acharya2026passive} and Theroem~\ref{thm:shunt_local_stab}.
    }
\label{fig:heatmaps}
\end{figure}

% \textbf{Shifted output.}
% Fix a constant reference $y_p\in\mathbb{R}$ and define the shifted output
% \begin{equation}
% \Delta y \coloneqq y - y_p.
% \label{eq:y_sh_def}
% \end{equation}
% In this paper, motivated by our previous work we take
% \begin{equation}
% y_p = -0.75.
% \label{eq:y0_value}
% \end{equation}

% \textbf{Passive-shunt feedback (PSF).}
% We apply the static output feedback law
% \begin{equation}
% u = u_p \coloneqq -k\, \Delta y
% = -k\,(y-y_p),
% \qquad k>0.
% \label{eq:shunt_feedback}
% \end{equation}
% We fix $k=1$ in the analysis below.

Motivated by this numerical observation, the following result proves the asymptotic stability of the full nonlinear model for a nominal combination of $u^\star$ and $k$. The same proof can be repeated, generally with different $V(\xbf)$, for other stabilizing combinations of $u^\star$ and $k$ as well.

\begin{theorem}\longthmtitle{Asymptotic stabilization under passive feedback}
\label{thm:shunt_local_stab}
Consider the Epileptor model~\eqref{eq:dyn}-\eqref{eq:stdc1c2} with output $y=x_1-x_2$
and feedback~\eqref{eq:shunt_feedback}. For
\begin{align}\label{eq:nominalparams}
u^\star = -0.8, \qquad \text{and} \qquad k = 1,
\end{align}
the closed-loop system is locally asymptotically stable towards the equilibrium point
\begin{align}\label{eq:xs_star_value}
\xbf^\star =
\begin{bmatrix}
-1.03 & -4.33  & -1.08 & 0 & -0.10 & 2.27
\end{bmatrix}^\top,
% \begin{bmatrix}
% -1.0329 \\
% -4.3349 \\
% 2.2682 \\
% -1.0829 \\
% 0 \\
% -0.1033
% \end{bmatrix},
\end{align}
with a region of attraction that includes (at least) the ball
\begin{align}\label{eq:ball}
\Bc = \big\{\xbf \in \real^n \ \big| \ \|\xbf - \xbf^\star\| < r_0 \big\},
\end{align}
with $r_0 = 0.56$.
\end{theorem}

\begin{proof}
It is straightforward to check that $\xbf^\star$ in~\eqref{eq:xs_star_value} satisfies
\begin{align*}
    \zeros = f(\xbf^\star) + \gbf u^\star,
\end{align*}
and is thus an equilibrium.
%
% Consider the closed-loop system obtained by applying the PSF law
% $u_p = -(y-y_p)$
% to \eqref{eq:affine}.
% Any equilibrium of this system satisfies $0=f(\xbf)+g\,u_p(\xbf)$.
% Solving this equation yields the closed-loop equilibrium $\xbf_s^\star$
% given in \eqref{eq:xs_star_value}.
%
To assess the local stability of the nonlinear closed-loop system around $\xbf^\star$, note that its linearization is given by, in shifted coordinates~\eqref{eq:tildes}, 
% Define shifted coordinates $\tilde{\xbf}=\xbf-\xbf_s^\star$ and linearize the
% closed-loop dynamics about $\xbf_s^\star$ to obtain
\begin{align}\label{eq:linearized-cl}
\dot{\tilde \xbf} = \big(\Abf - \gbf \cbf^\top\big) \tilde \xbf,
\end{align}
% \dot{\tilde{\xbf}} = A_{\mathrm{cl}}\tilde{\xbf},\qquad
% A_{\mathrm{cl}}
% =
% \left.\frac{\partial}{\partial \xbf}\bigl(f(\xbf)+g\,u_p(\xbf)\bigr)\right|_{\xbf_s^\star}
% =
% A_s - g\,c_{\mathrm{std}},
where $\Abf=\frac{\partial f}{\partial \xbf}\Big|_{\xbf^\star}$ and
$\cbf$ is in~\eqref{eq:stdc},~\eqref{eq:stdc1c2}. %the output-gradient corresponding to $y=x_1-x_2$. 
The matrix $\Abf - \gbf \cbf^\top$ is Hurwitz (cf. Fig.~\ref{fig:heatmaps}), and thus linearized closed-loop system~\eqref{eq:linearized-cl} is asymptotically stable. 

Since $\Abf - \gbf \cbf^\top$ is Hurwitz, there exists a quadratic Lyapunov
function $V(\tilde{\xbf})=\tilde{\xbf}^\top \Pbf \tilde{\xbf}$ with
$\Pbf \succ \zeros$ such that
\begin{align}\label{eq:lyap-eq-cl}
\big(\Abf - \gbf \cbf^\top\big)^\top \Pbf + \Pbf \big(\Abf - \gbf \cbf^\top\big) \prec \zeros.
\end{align}
One such Lyapunov function among the infinitely many that satisfy~\eqref{eq:lyap-eq-cl}, is given by
\begin{align*}
V(\tilde \xbf) = 1.1 \tilde x_1^2 &+ 0.12 \tilde y_1^2 + 0.05 \tilde x_2^2 + 59.3 \tilde y_2^2 + 364 \tilde \zeta^2 
\\
&+ 848 \tilde z^2 + 0.14 \tilde y_1 \tilde x_2
\end{align*}
which is obtained by solving the
sparsity-promoting convex program
\begin{align*}
&\min_{\Pbf \succeq \delta \Ibf} && \sum_{i \neq j} |P_{ij}|
% \|\Pbf\|_{1,\mathrm{off}}
\\ 
&\ \ \text{s.t.} &&
\big(\Abf - \gbf \cbf^\top\big)^\top \Pbf + \Pbf \big(\Abf - \gbf \cbf^\top\big) \preceq -\varepsilon \Ibf,
\end{align*}
with $\varepsilon = 10^{-4}$ and $\delta = 10^{-6}$. 
% fixed $\varepsilon,\delta>0$. Here
% \[
% \|P\|_{1,\mathrm{off}} := \sum_{i\neq j} |P_{ij}|
% \]
% denotes the $\ell_1$ norm of the off-diagonal entries of $P$, which
% promotes sparsity.

% The resulting Lyapunov matrix $P$ is highly sparse. In the coordinates
% $\tilde{\xbf}=[\tilde{x}_1,\tilde{y}_1,\tilde{z},\tilde{x}_2,\tilde{y}_2,\tilde{\zeta}]^\top$,
% the only nonzero entries are the diagonal elements
% $P_{11},P_{22},P_{33},P_{44},P_{55},P_{66}$ and the symmetric coupling
% $P_{24}=P_{42}$; all other entries are zero. The corresponding polynomial
% parameterization of $V(\tilde{\xbf})$ is reported in
% Table~\ref{tab:lyap_params}.

% \begin{table}[t] \centering \renewcommand{\arraystretch}{1.2} \caption{Quadratic Lyapunov function parameterization $V(\tilde{x})=\tilde{x}^\top P \tilde{x}$ } \label{tab:lyap_params} \begin{tabular}{lcl} \hline Parameter & Monomial term in $V(\tilde{x})$ & Value \\ \hline $P_{11}$ & $\tilde{x}_1^2$ & $0.0112$ \\ $P_{22}$ & $\tilde{y}_1^2$ & $0.0012$ \\  $P_{33}$ & $\tilde{x}_2^2$ & $0.0005$ \\ $P_{44}$ & $\tilde{y}_2^2$ & $0.5934$ \\ $P_{55}$ & $\tilde{\zeta}^2$ & $3.6386$ \\$P_{66}$ & $\tilde{z}^2$ & $8.4781$ \\ $P_{23}=P_{32}$ & $2\,\tilde{y}_1 \tilde{x}_2$ & $0.0007$ \\ \hline \end{tabular} \end{table}

% \paragraph*{Nonlinear certificate via SOS.}

To obtain the region of attraction of $\xbf^\star$ in the full nonlinear closed-loop system, we solved the following convex sum of squares (SOS) program
\begin{subequations} \label{eq:sostools}
\begin{align}
    &\max && \rho
    \\
    &\ \ \text{s.t.} && V(\tilde \xbf) \le \rho \;\Rightarrow\;
    \\
    &&& \nabla V(\tilde \xbf)^\top \Big(f(\xbf^\star + \tilde \xbf) + \gbf u^\star - \gbf \cbf^\top \tilde \xbf\Big) \le -\epsilon \|\tilde \xbf\|^2
\end{align}
\end{subequations}
with $\epsilon = 10^{-6}$ using~\cite{Prajna2002SOSTOOLS}, resulting in the bound
$
    \rho^\star = 1.66561 \times 10^{-3}.
$
Since $V(\xbf) \ge \lambda_{\min}(\Pbf) \|\tilde \xbf\|^2$, it follows that the sublevel set
\begin{align*}
\Omega_{\rho^\star} = \big\{\xbf \in \real^n \ \big| \ V(\tilde \xbf) \le \rho^\star \big\}
\end{align*}
contains the Euclidean ball~\eqref{eq:ball} with radius
\begin{align*}
r_0 = \sqrt{\rho^\star / \lambda_{\min}(\Pbf)} \simeq 0.56,
\end{align*}
completing the proof.
\end{proof}

The simulated trajectories of the closed-loop system at the nominal feedback parameters~\eqref{eq:nominalparams} are shown in Fig.~\ref{fig:sims}. The same proof as in Theorem~\ref{thm:shunt_local_stab} can also be stated, albeit with different $V(\xbf)$, for other stabilizing combinations of $u^\star$ and $k$ in Fig.~\ref{fig:heatmaps}.

% \begin{remark}
%  Although the standard sensing
% configuration violates the Positive Real Lemma,
% static output feedback of passive form still achieves
% local stabilization (illustrated in Fig.~\ref{fig:sims}A-B).
% \end{remark}

\begin{figure}[t]
\centering
\subfloat[Simulated trajectory]{%
\includegraphics[width=0.45\linewidth]{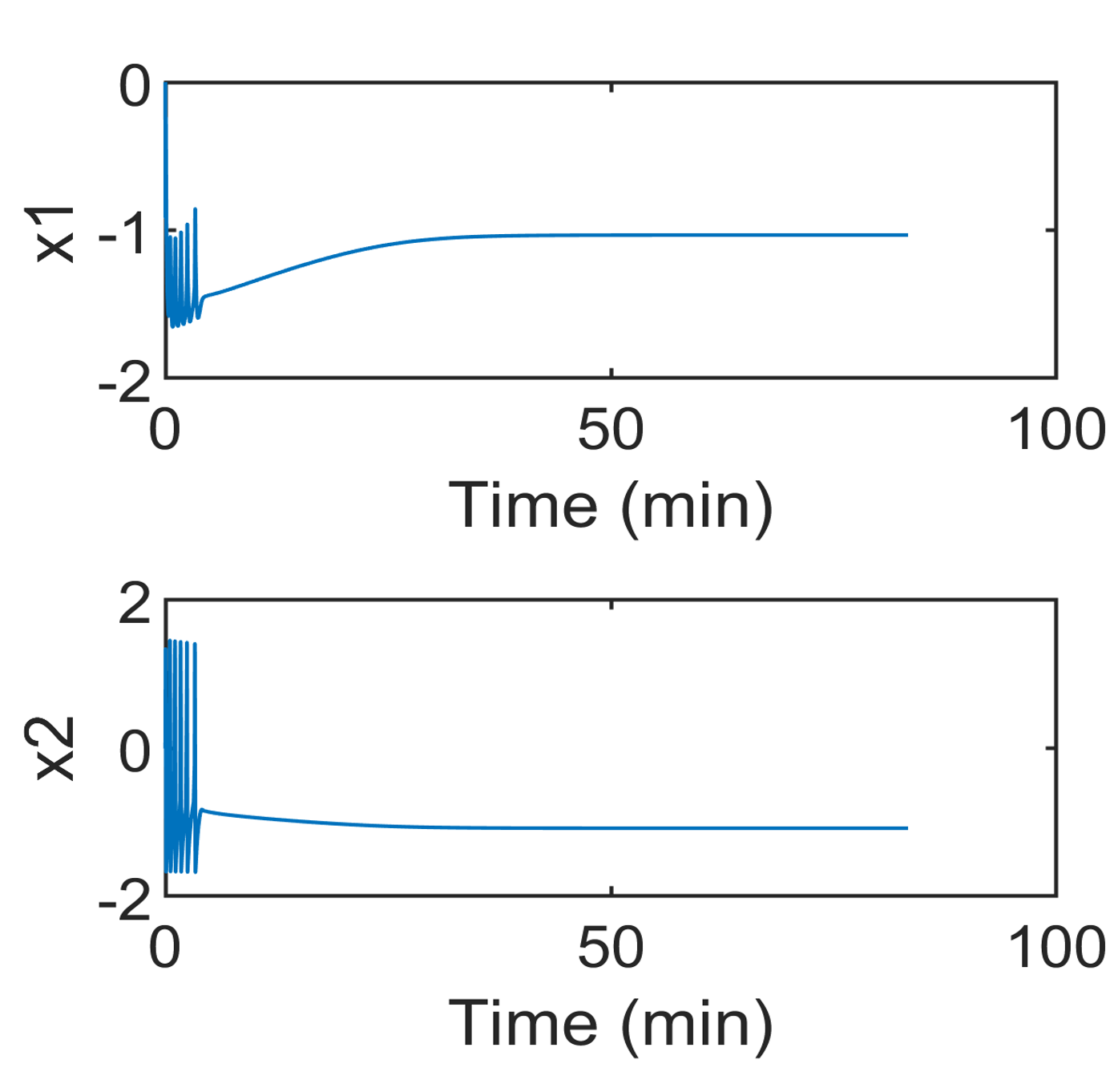}}
\hfill
\subfloat[Phase plane $(x_2,y_2)$]{%
\includegraphics[width=0.53\linewidth]{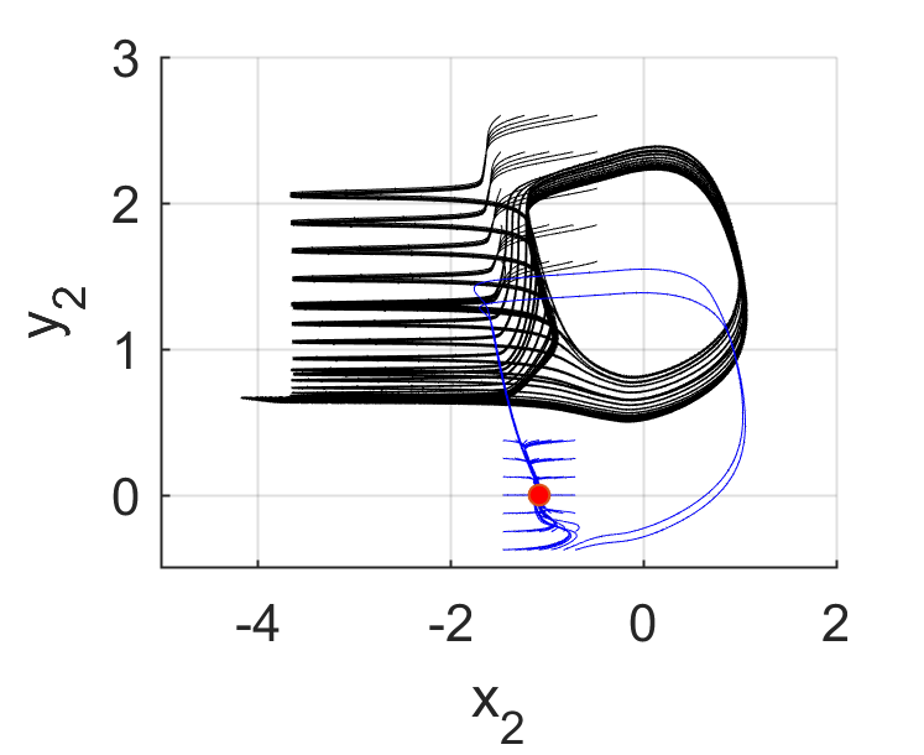}}
\caption{Closed-loop behavior of the Epileptor under passive feedback with parameters as in Theorem~\ref{thm:shunt_local_stab}. Left: asymptotic convergence of two representative states to steady-state values. Right: the phase plane of the slow subsystem, showing the co-existence of the asymptotically stable equilibrium point (red dot) and a stable limit cycle (black). The apparent overlap between the regions of attraction of the stable equilibrium and the limit cycle is due to the 2D projection of the 6D state space.}
\label{fig:sims}
\end{figure}

\begin{remark}\longthmtitle{Lack of global stability}
The feedback stabilization achieved by~\eqref{eq:fbform}-\eqref{eq:nominalparams} is only local. Due to the intrinsic
slow–fast bistability of the Epileptor,
the closed-loop system retains a coexisting
stable limit cycle that is attractive for certain initial conditions.
Fig.~\ref{fig:sims}(b)
illustrates this coexistence of a stable equilibrium and a stable limit cycle. \oprocend %trajectories converging either to
% $x^\star$ or to a persistent oscillatory attractor,
% demonstrating that seizure oscillations are not
% globally eliminated.
\end{remark}

Given the ability of passive feedback to stabilize the system, one might hypothesize that the resulting closed-loop system~\eqref{eq:dyn}-\eqref{eq:stdc1c2},~\eqref{eq:fbform} also becomes (strictly) passive with respect to the incremental input $v$. This is not the case, however, as discussed next.

%%%%%%%%%%%%%%%%%%%%%%%%%%%%%%%%%%%%%%%%%%%%%%%%%%%%%%%%%%%%%%%%%%%%%%%%%%%%%%%%
\section{Passivation of the Epileptor via Output Redesign}
\label{sec:motiv_output}

\subsection{Passivation Under the Standard Output}\label{subsec:passivation-std}

In this section, we tackle the more ambitious objective of Problem~\ref{prob}b, seeking to not only stabilize but also passivate the Epileptor via feedback.
As seen both from~\eqref{eq:matching}-\eqref{eq:dissipation} in the nonlinear case and~\eqref{eq:PRLb}-\eqref{eq:PRLc} for linear systems, passivity of input-affine systems fundamentally consists of an internal energy dissipation and an input-output matching condition. The following result highlights a necessary structural restriction that the latter imposes on the Epileptor.

% The following result isolates a necessary structural condition for passivity imposed by input--output matching.

\begin{theorem}\longthmtitle{Matching obstruction for passivity of the standard Epileptor model}\label{thm:matching_obstruction}
Assume that the Epileptor model in~\eqref{eq:dyn}-\eqref{eq:y} is passive around an equilibrium point $\xbf^\star$ with a twice continuously-differentiable storage function $V(\xbf)$. Then we must have
\begin{align}\label{eq:c1c2pos}
    c_1 + c_3 > 0.
\end{align}
% with $c\neq 0$ and
% \begin{equation}
% g =
% \begin{bmatrix}
% g_1 & 0 & 0 & g_4 & 0 & 0
% \end{bmatrix}^\top,
% \qquad (g_1,g_4)\neq (0,0).
% \label{eq:g_structure}
% \end{equation}
% If \eqref{eq:lin} is passive, then necessarily
% \begin{equation}
% g_1 c_1 + g_4 c_4 > 0,
% \label{eq:gc_positive}
% \end{equation}
% where $c_i$ denotes the $i$th entry of $c$.
\end{theorem}

\begin{proof}
Since $V(\xbf)$ is twice continuously differentiable around $\xbf^\star$, it can be expanded as a Taylor series
\begin{align*}
V(\xbf) &= V(\xbf^\star) + \nabla V(\xbf^\star)^\top \tilde \xbf + \frac{1}{2} \tilde \xbf^\top \Pbf \tilde \xbf + R(\tilde \xbf)
\end{align*}
where $R(\tilde \xbf) = \Oc(\|\tilde \xbf\|^3)$, $\tilde \xbf = \xbf - \xbf^\star$ as in~\eqref{eq:equilibrium}, and $\Pbf = \nabla^2 V(\xbf^\star)$. Since $V(\xbf)$ is positive semidefinite, $V(\xbf^\star) = 0$, $\nabla V(\xbf^\star) = \zeros$, and $\Pbf \succeq \zeros$. Thus, $\nabla V$ simplifies to
\begin{align*}
\nabla V(\xbf) &= \Pbf \tilde \xbf + \nabla R(\tilde \xbf).
% \\
% \nabla R(\tilde \xbf) &= \Oc(\|\tilde \xbf\|^2).
\end{align*}
Substituting this into~\eqref{eq:matching} we obtain
\begin{align*}
\tilde \xbf^\top \Pbf \gbf + \nabla R(\tilde \xbf)^\top \gbf = \tilde \xbf^\top \cbf
\end{align*}
where $\cbf$ is as in~\eqref{eq:stdc}. Since this must hold for all $\xbf$ in a neighborhood of $\xbf^\star$ and $\nabla R(\tilde \xbf)^\top \gbf = \Oc(\|\tilde \xbf\|^2)$, we must have $\nabla R(\tilde \xbf)^\top \gbf = 0$ and equating first-order terms yields
% By Proposition~1 (Positive Real Lemma), passivity of \eqref{eq:lin} implies the existence of a symmetric matrix $P=P^\top\succeq 0$ such that
\begin{equation}\label{eq:PRL_match_here}
\Pbf \gbf = \cbf.
\end{equation}
Left-multiplying~\eqref{eq:PRL_match_here} by $\gbf^\top$ gives
\begin{equation}\label{eq:cg_equals_quad}
\gbf^\top \Pbf \gbf = \gbf^\top \cbf.
\end{equation}
Since $\Pbf \succeq \zeros$, we have $\gbf^\top \Pbf \gbf \ge 0$. Moreover, $\gbf^\top \Pbf \gbf = 0$ implies $\Pbf \gbf = 0$ which, by~\eqref{eq:PRL_match_here}, forces $\cbf = \zeros$, contradicting~\eqref{eq:stdc}. Therefore, $\gbf^\top \Pbf \gbf > 0$, and thus
\begin{equation}
\gbf^\top \cbf = c_1 + c_3 > 0,
\end{equation}
completing the proof.
\end{proof}

Theorem~\ref{thm:matching_obstruction} identifies a necessary structural condition for passivating the Epileptor seizure dynamics. It can be seen from the proof that the condition~\eqref{eq:c1c2pos} depends only on $c_1$ and $c_3$ \textit{not} because of the sparse choice of the output~\eqref{eq:y}, but rather because of the sparse nature of $\gbf$. Importantly, \eqref{eq:c1c2pos} would still remain necessary if other (linear) terms are added to the output (cf. Section~\ref{subsec:opt}).

% \erfan{We can put a remark or paragraph here emphasizing that even if we had other linear terms in the output, only $c_1 + c_3$ would matter because of the form of $\gbf$.}

\begin{corollary}\longthmtitle{Matching obstruction for feedback passivation of the standard Epileptor}
\label{cor:std_not_passive}
Consider the standard Epileptor model~\eqref{eq:dyn}-\eqref{eq:stdc1c2} under the feedback input~\eqref{eq:fbform}. The resulting closed-loop system (with $v$ as input) cannot be passive for any choice of $\phi$.
% Let $g = [1,0,0,1,0,0]^\top$ and $c=c_{\mathrm{std}}=[1,0,0,-1,0,0]$. Then the matching condition \eqref{eq:PRL_match_here} cannot hold, and the linearization \eqref{eq:lin} is not passive. Consequently, the Epileptor model is not locally passive about $\xbf^\star$ under the standard LFP proxy output \eqref{eq:ystd}.
\end{corollary}

\begin{proof}
From~\eqref{eq:affine}, \eqref{eq:fbform}, and~\eqref{eq:y} the resulting closed-loop system still has the input-affine form
\begin{align*}
\dot \xbf &= \hat f(\xbf) + \gbf v
\\
y &= \cbf^\top \xbf
\end{align*}
where $\hat f(\xbf) = f(\xbf) + \gbf u^\star - \gbf \phi(\cbf^\top \xbf - y^\star)$. Therefore, the closed-loop system must still satisfy the necessary condition of Theorem~\ref{thm:matching_obstruction} to be passive, which is violated by the standard output~\eqref{eq:stdc1c2}. 
% and is thus subject to the same necessary condition
% Here $c g = 1\cdot 1 + 1\cdot (-1)=0$, violating the necessary condition \eqref{eq:gc_positive}. Thus no $P\succeq 0$ can satisfy $Pg=c^\top$, and passivity fails by the Positive Real Lemma.
\end{proof}

As noted earlier, the passivity of input-affine systems decouples into an internal dissipation and an input-output matching. The above results show that, while the Epileptor can be robustly stabilized by passive feedback, it cannot be made passive due to a failure of the input-output matching condition.

% This obstruction is purely algebraic and holds regardless of the linearized drift matrix $A$.

% In Section~\ref{sec:passivity_standard}, we showed that the standard LFP proxy output
% $y = x_1 - x_2$ fails the algebraic matching condition required
% by the Positive Real Lemma. In particular, Theorem~1 showed
% that for scalar outputs $y = c x$, passivity of the linearization
% requires \eqref{eq:gc_positive}. For the standard choice $c = [1,0,0,-1,0,0]$, this inequality 
% fails, creating a structural obstruction to passivity.

% We now investigate whether appropriate output redesign can restore passivity and thereby enable passive stabilization.

\subsection{Passivation via Output Redesign}\label{subsec:yc1pc2}

To address the matching obstruction shown in Section~\ref{subsec:passivation-std}, we need to change the output such that the condition~\eqref{eq:c1c2pos} is satisfied. A rather straightforward and interpretable
%
% For the Epileptor input configuration
% $g_1 = g_4 = 1$, \eqref{eq:gc_positive} becomes
% \begin{equation}
% c_1 + c_4 > 0.
% \label{eq:nominal_matching}
% \end{equation}
%
% A particularly symmetric and physically interpretable
choice is
\begin{equation}\label{eq:symmetric_choice}
c_1 = c_3 = 1,
\end{equation}
corresponding to sensing the sum of the fast and the slow
subsystems (cf. Section~\ref{sec:prelim_ep}). The following result proves the passivity of the resulting closed-loop system at a nominal setting of feedback parameters.

% We now analyze whether this output restores passivity.

% \subsection{Passivity Under Symmetric Output}

% As discussed in Section~\ref{sec:prelim_ep}, the autonomous
% system ($u=0$) does not admit an asymptotically stable equilibrium
% in the parameter regime considered here.
% To enable equilibrium stabilization, we therefore introduce
% a constant bias input $u_e$ and study the closed-loop
% dynamics about the corresponding forced equilibrium.

% For sufficiently negative constant inputs,
% the vector field $f(\x)+g u_e$ admits a stable equilibrium.
% Numerical calculation reveals that for
% \[
% u_e \coloneqq u^\star_e = -2,
% \]
% the system possesses a single stable equilibrium
% \[
% \x^\star_e =
% \begin{bmatrix}
% -1.3706 \\
% -8.3926 \\
% 0.9176 \\
% -1.3362 \\
% 0 \\
% -0.1371
% \end{bmatrix},
% \]

% Let us
% consider the linearization about ($\x_e^\star$,$u_e^\star$) in~\eqref{eq:lin} with $A\coloneqq A_e=\left.\frac{\partial f}{\partial \xbf}\right|_{\xbf_e^\star}$ and
%  $c\coloneqq c_{new} = [1,0,0,1,0,0]$, motivated from our discussion in \ref{subsec:motiv_output}. 

\begin{theorem}\longthmtitle{Strict passivity following output redesign}\label{thm:passive-newc1c2}
Consider the Epileptor model~\eqref{eq:dyn}-\eqref{eq:y} with the choice of output in~\eqref{eq:symmetric_choice}. Let $u$ be given by the feedback form~\eqref{eq:fbform},~\eqref{eq:linear-phi} and
\begin{align}\label{eq:nominal-with-k-0}
u^\star = -2, \qquad \text{and} \qquad k = 0.
\end{align}
Then, the resulting closed-loop system is locally strictly passive about $\xbf^\star$; in particular, the passivity inequality holds for all $u \in \mathbb{R}$ and all $\xbf$ in the ball

\begin{align*}
\Bc := \big\{\xbf \in \mathbb{R}^n \ \big| \ \|\xbf-\xbf^\star\|_2 \le r_1\big\},
\end{align*} 
around the equilibrium point
\begin{align}\label{eq:xstar-um2}
\xbf^\star =
\begin{bmatrix}
-1.37 & -8.39 & -1.34 & 0 & -0.14 & 0.92
\end{bmatrix}^\top,
\end{align}
with radius $r_1 = 1.008$.
\end{theorem}

\begin{proof}
It is straightforward to see that~\eqref{eq:PRLa}-\eqref{eq:PRLc} are satisfied strictly for the linearized system~\eqref{eq:lin} and the diagonal matrix
\begin{align*} 
\Pbf = \operatorname{diag}\big(
\begin{bmatrix} 
1 & 0.074 & 1 & 125 & 143 & 600
\end{bmatrix}\big).
\end{align*}
Therefore, the linearized system~\eqref{eq:lin} is strictly passive by Proposition~\ref{prop:prl}, and the local strict passivity of the nonlinear system follows from Proposition~\ref{prop:local-suff}. The radius $r_1 = 1.008$ is obtained via the same sum of squares (SOS) Lyapunov analysis as in the proof of Theorem~\ref{thm:shunt_local_stab}. %bounding the higher-order nonlinear terms in the passivity inequality and computing, via SOS certification, the largest Euclidean ball centered at $\xbf^\star$ within which the strict passivity condition holds.
% 
% The following $P$ matrix was indentified numerically using YALMIP and satisfies~\eqref{eq:pass_out_select}
\end{proof}

The nominal values of the feedback parameters in~\eqref{eq:nominal-with-k-0} were chosen for simplicity based on Fig.~\ref{fig:heatmaps}, where $u^\star = -2$ is sufficient to stabilize the system without further feedback ($k = 0$). Nevertheless, the same numerical sweep as in Fig.~\ref{fig:heatmaps} can be performed for the new choice of output~\eqref{eq:symmetric_choice}, the results of which are shown in Fig.~\ref{fig:heatmaps-new}a. Note that, compared to Fig.~\ref{fig:heatmaps}, the range of feedback parameters that can stabilize the system has significantly shrunk. This highlights the presence of a tradeoff between stabilization and passivation, where redesigning the output for the latter can deteriorate the former. In the following Subsection, we address this issue, as well as the general arbitrariness of the choice in~\eqref{eq:symmetric_choice}.

\begin{figure}[t]
\centering
\subfloat[]{%
\includegraphics[width=0.5\linewidth]{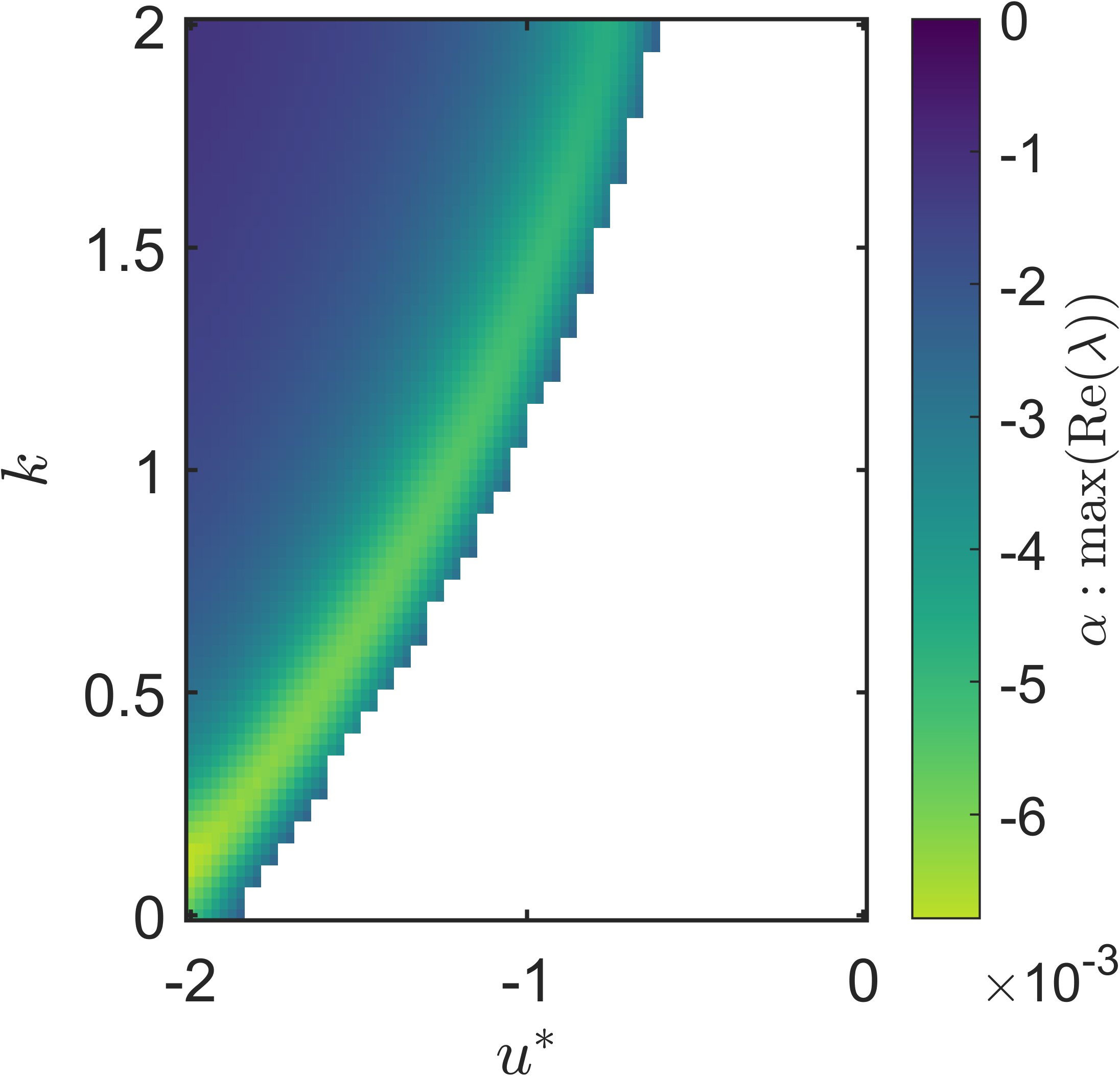}}
\hfill
\subfloat[]{%
\includegraphics[width=0.5\linewidth]{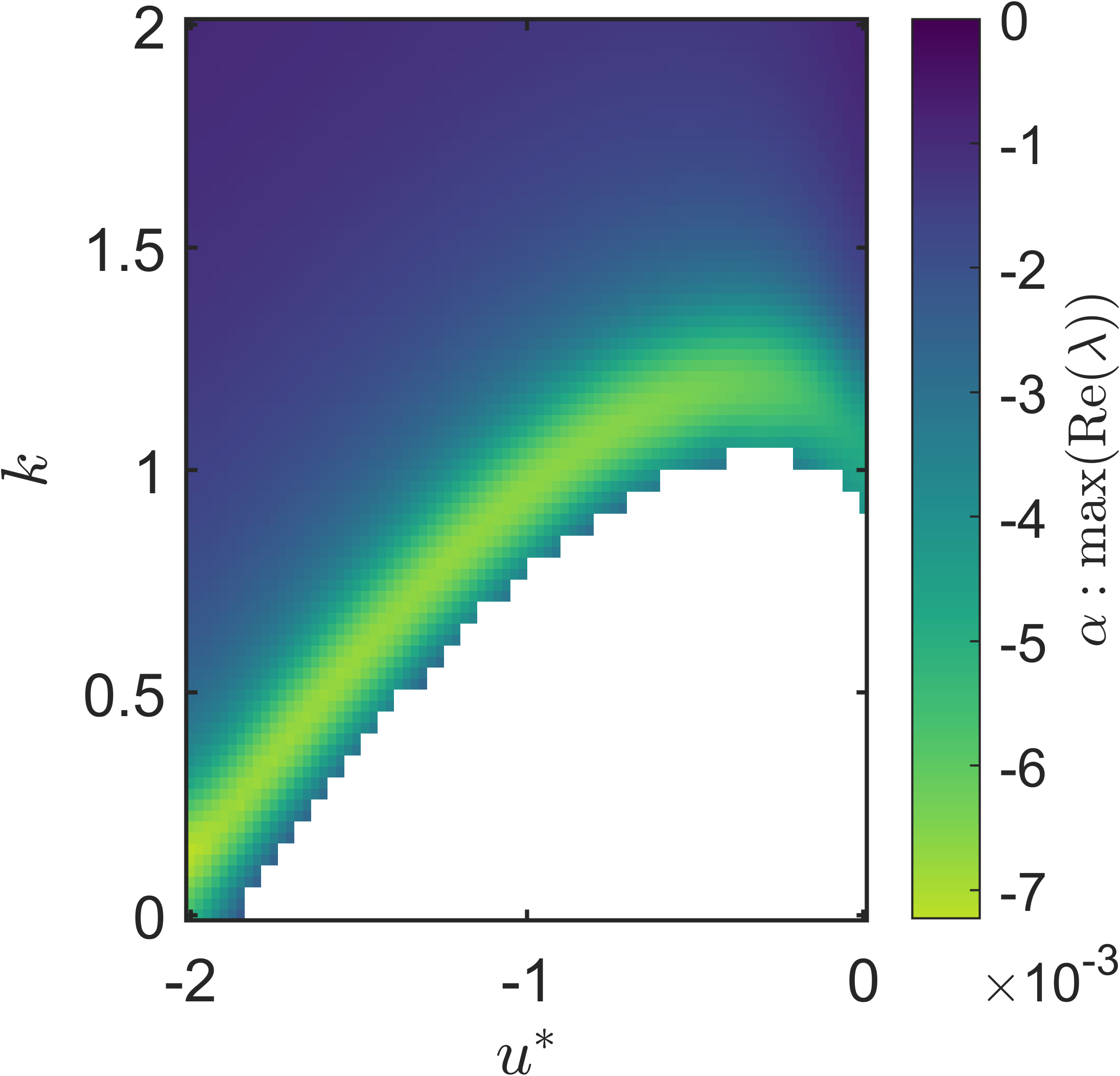}}
\caption{The spectral abscissa of the closed-loop Jacobian matrix $\Abf - k \gbf \cbf^\top$, similar to that in Fig.~\ref{fig:heatmaps}, but for new sets of controller and output parameters. (a) $y = x_1 + x_2$, $u^\star = -2$, $k = 0$ (Theorem~\ref{thm:passive-newc1c2}), (b) $y = x_1 - 0.29 x_2$, $u^\star = -2$, $k = 0$ (Example~\ref{ex:opt}).}
\label{fig:heatmaps-new}
\end{figure}

\subsection{Optimal Sensor Placement}\label{subsec:opt}

As noted in Section~\ref{sec:prelim_ep}, the standard output~\eqref{eq:stdc1c2} was motivated by the similarity of $y = x_1 - x_2$ to empirically recorded EEG in patients with epilepsy. Accordingly, any change in the output, such as~\eqref{eq:symmetric_choice}, models an empirical change in the number and location of sensors used to record brain activity. Therefore, while infinitely many choices of $c_1$ and $c_3$ can satisfy Theorem~\ref{thm:matching_obstruction}, they may differ significantly in terms of both their feasibility and optimality. In this section we provide an optimization framework that allows us to design a general linear output
\begin{align}\label{eq:gen-y}
y = \cbf^\top \xbf, \qquad \cbf = \begin{bmatrix}
    c_1 & c_2 & c_3 & c_4 & c_5 & c_6
\end{bmatrix}^\top,
\end{align}
such that the closed-loop system is not only strictly passive, but also optimal according to a user-specified objective
\begin{align}\label{eq:ell}
\min_{\cbf \in \real^6} \ell(\cbf).
\end{align}
The loss function $\ell(\cdot)$ can capture, e.g., the sparsity of $\cbf$ or its distance from an empirically-feasible output.

% address this fact by casting the output redesign as a convex optimization problem.

Consider again the Epileptor model~\eqref{eq:dyn}-\eqref{eq:y} under the passive feedback~\eqref{eq:fbform},~\eqref{eq:linear-phi}. Assume, for simplicity, that the feedback parameters $u^\star$ and $k$ are fixed such that the closed-loop system with $v = 0$ is asymptotically stable. 
The linearization of the closed-loop system around the equilibrium $\xbf^\star$ imposed by $u^\star$ is then given by
\begin{subequations}\label{eq:cl-lin}
\begin{align}
\dot{\tilde \xbf} &= (\Abf - k \gbf \cbf^\top) \tilde \xbf + \gbf v,
\\
\tilde y &= \cbf^\top \tilde \xbf.
\end{align}
\end{subequations}
Substituting~\eqref{eq:cl-lin} in Proposition~\ref{prop:prl} and combining it with~\eqref{eq:ell} results in the following optimization problem.

\begin{problem}\label{prob:opt}
Given $\Abf \in \real^{n \times n}$, $k \ge 0$, $\gbf \in \real^n$, and $\ell: \real^n \to \real$, find $\Pbf = \Pbf^\top \in \real^{n\times n}$ by solving
\begin{subequations}
\begin{align*}
\min_{\Pbf} \quad & \ell(\Pbf \gbf) 
\\
\text{s.t.} \quad & \Abf^\top \Pbf + \Pbf \Abf - 2 k \Pbf \gbf \gbf^\top \Pbf \preceq \zeros, 
\\
& \Pbf \succeq \zeros. \eqoprocend
% \\
% & c = (Pg)^\top ,
\end{align*}
\end{subequations}
% and let $\cbf = \Pbf \gbf$. \oprocend
\end{problem}

Once Problem~\ref{prob:opt} is solved, the desired output vector can be found by $\cbf = \Pbf \gbf$. Any solution to Problem~\ref{prob:opt} satisfies~\eqref{eq:PRLa}-\eqref{eq:PRLc} by construction, and the passivity of the full nonlinear system can be guaranteed similar to Theorem~\ref{thm:passive-newc1c2}. 

\begin{remark}\longthmtitle{Convexity of Problem~\ref{prob:opt}}
Despite the familiar Riccati-like form of the constraint in Problem~\ref{prob:opt}, it is in general non-convex. However, if $\ell(\cdot)$ is convex and $k = 0$ the problem becomes convex. Note that this requires the baseline input $u^\star$ to be such that it alone suffices for asymptotic stabilization of~\eqref{eq:dyn}. Sufficiently negative $u^\star$ can achieve this, as can be seen from Fig.~\ref{fig:heatmaps}. \oprocend
\end{remark}

% In the previous subsection we motivated the sensing vector
% $c = [1\;0\;0\;1\;0\;0]$ and showed that it yields local passivity.
% More generally, if the sensing configuration is not fixed, one may
% seek other sparse outputs that satisfy the same passivity relation.

% Suppose the pair $(A,g)$ is given, with $A$ being Hurwitz and the output map $y=cx$ is to be
% designed. We know that exists a Lyapunov matrix such that $P \succ 0$ and $A^\top P + P A \prec 0$. If we set the sensing vector $c\coloneqq c_{P,g}= (Pg)^\intercal$, then the system $(A,g,c_{P,g})$ is bound to be passive. 

% To promote sparse sensor placement, we search for a matrix $P$
% that induces a sparse vector $Pg$. This leads to the convex program
% \begin{align}
% \begin{aligned}
% \min_{P,c} \quad & \Phi(c) \\
% \text{s.t.}\quad 
% & A^\top P + P A \preceq 0, \\
% & P \succeq 0, \\
% & c = (Pg)^\top ,
% \end{aligned}
% \end{align}
% where $\Phi(c)$ is a convex function promoting sparsity in the sensing vector.

\begin{example}\longthmtitle{Feedback passivation with minimal deviation from standard output.}\label{ex:opt}
Consider the Epileptor model~\eqref{eq:dyn} with the general output~\eqref{eq:gen-y} and the passive feedback~\eqref{eq:fbform},~\eqref{eq:linear-phi}. We seek to design $\cbf$ such that $y$ is as close to the standard Epileptor output~\eqref{eq:stdc1c2} as possible while the closed-loop system is strictly passive. To achieve this, we can solve Problem~\ref{prob:opt} with $u^\star = -2$, $k = 0$, and the $\ell_1$ loss
\begin{align*}
\ell(\cbf) &= \big\|\cbf - \begin{bmatrix}
    1 & 0 & -1 & 0 & 0 & 0
\end{bmatrix}^\top\big\|_1
\\
&= |c_1 - 1| + |c_2| + |c_3 + 1| + |c_4| + |c_5| + |c_6|. 
\end{align*}
% In particular, we seek sensing configurations that depend
% primarily on the fast states $x_1$ and $x_2$, we encourage the remaining
% entries of $c$ to vanish. This is achieved by selecting
% \begin{equation}
%     \Phi(c) = \|[c_2,c_3,c_5,c_6]\|_1 + \psi(c_1,c_4),
% \end{equation}
% where the $\ell_1$ term promotes sparsity in the undesired sensor
% locations, driving $c_2,c_3,c_5,c_6$ toward zero, while
% $\psi(\cdot,\cdot)$ is a convex function that can encode desired
% structure between the remaining coefficients. For instance,
% $\psi(c_1,c_4)$ may encourage similar weights on the $x_1$ and $x_2$
% states, reflecting the symmetry between the two fast subsystems.
% 
% 
% The resulting $c$ defines an output map $y=cx$ that preserves
% the passivity structure while encouraging sparse sensor placement.
% 
% \subsection{Relaxed sensing structure}
% \label{subsec:relax_sense}
% We now specialize the objective to explore sensing configurations that
% remain close to the standard output $y = x_1 - x_2$ while allowing some
% flexibility in the relative weighting of the two fast states. In
% particular, we consider outputs of the form
% \[
% y = x_1 - \beta x_2,
% \]
% where the coefficient $\beta>0$ is not fixed a priori.
% 
% To encourage this structure, we choose the convex penalty
% \[
% \psi(c_1,c_4) = |c_1 - 1| + |c_4 + 1|,
% \]
% which biases the optimization toward the nominal coefficients
% $(c_1,c_4)=(1,-1)$ while still permitting deviations if required by the
% passivity constraints.
% 
% Applying the optimization procedure to the system $(A_e,g)$ yields the
% sensing vector
The solution to Problem~\ref{prob:opt} then yields
$
c = \begin{bmatrix}
    1 & 0 & -0.29 & 0 & 0 & 0
\end{bmatrix},
$
and the resulting output
\begin{align}\label{eq:opt-y}
y = x_1 - 0.29x_2,
\end{align}
gives the closest choice to~\eqref{eq:stdc1c2} such that the closed-loop system is strictly passive.
% This result indicates that passive sensing can be achieved using a
% relaxed version of the standard output, with a significantly smaller
% weight on the $x_2$ state.

% \subsection{Stabilization under relaxed sensing}

% The sensing vector obtained in \ref{subsec:relax_sense} corresponds to the relaxed output
% \begin{equation}
% \label{eq:rel_out}
% y = x_1 - 0.2867\,x_2 .
% \end{equation}

% Let
% \[
% y_e^\star = x_{e,1}^\star - 0.2867\,x_{e,2}^\star
% \]
% denote the steady output at the equilibrium $x_e^\star$.
% We consider the shifted passive feedback law
% \begin{equation}
% \label{eq:new_input}
%     u= u_e^\star - (y - y_e^\star),
% \end{equation}
% 
% where $u_e^\star = -2$ is the equilibrium input.

Analytically, the strict passivity of the resulting system can be shown similar to Theorem~\ref{thm:passive-newc1c2} with the storage function $V(\xbf) = \frac{1}{2} \tilde \xbf^\top \Pbf \tilde \xbf$ where
% The passivity analysis developed earlier guarantees local
% stability of the linearized dynamics under this sensing
% configuration. To further assess the stability properties of the nonlinear closed-loop
% system, we again employ the nonlinear SOS certification procedure
% described earlier. First, a sparse quadratic Lyapunov function of the
% form
% \[
% V(x) = x^\top P x
% \]
% is computed for the closed-loop dynamics under the feedback law
% in~\eqref{eq:new_input}. The resulting Lyapunov matrix is
% diagonal,
\begin{align*}
\Pbf = \operatorname{diag}\big(\begin{bmatrix}
0.97 & 0.08  & 0.48 & 2.26 & 14.45 & 458 \end{bmatrix}\big).
\end{align*}
It can further be shown via sum of squares analysis, similar to Theorem~\ref{thm:passive-newc1c2}, that $V(\xbf)$ solves the optimization problem~\eqref{eq:sostools} around the same equilibrium $\xbf^\star$ as in~\eqref{eq:xstar-um2}, with a region of attraction that includes (at least) the ball
% Using this candidate Lyapunov function, we perform a sum-of-squares
% (SOS) feasibility search to certify negativity of $\dot V(x)$ for the
% full nonlinear model. The resulting certificate guarantees that
% $\dot V(x)<0$ within the sublevel set $V(x)\le \rho^\star$, which
% corresponds to a Euclidean ball of radius
\begin{align*}
\Bc = \big\{\xbf \in \real^n \ \big| \ \|\xbf - 
\xbf^\star\| < r_2 \big\},
\end{align*}
where $r_2 = 1.08$ is about twice as large as $r_0$ obtained for $y = x_1 - x_2$.

Furthermore, Fig.~\ref{fig:heatmaps-new}b shows the result of a similar 2D sweep of feedback parameter space for the optimal output~\eqref{eq:opt-y}. Interestingly, the closed-loop system is asymptotically stable for a wider range of $(u^\star, k)$ compared to both Fig.~\ref{fig:heatmaps} (standard output $y = x_1 - x_2$) and Fig.~\ref{fig:heatmaps-new}a (suboptimal output $y = x_1 + x_2$).
\oprocend
% This provides empirical evidence that the feedback law in~\eqref{eq:new_input} yields stronger stabilization than the PSF input introduced in~\ref{sec:passivity_standard}, which is based on the standard output formulation. Furthermore, numerical simulations of the full nonlinear system indicate convergence of the closed-loop trajectories to $\x_e^\star$ from a wide range of initial conditions. 
\end{example}

The results of Example~\ref{ex:opt} show the promise of Problem~\ref{prob:opt} for relaxing the tradeoff between stabilization and passivation we had observed in Section~\ref{subsec:yc1pc2}. This can be particularly valuable in the greater context of PBC for seizure suppression, whereby various empirical constraints may restrict the set of outputs that can be measured and used for PBC. 

%\begin{conjecture}[Global stabilization under relaxed sensing]
%Consider the Epileptor model with the relaxed output in~\eqref{eq:rel_out}
%and feedback in~\eqref{eq:new_input}.
%Then the resulting closed-loop system is globally
%asymptotically stable at the equilibrium $x_e^\star$. 
%\end{conjecture}
%%%%%%%%%%%%%%%%%%%%%%%%%%%%%%%%%%%%%%%%%%%%%%%%%%%%%%%%%%%%%%%%%%%%%%%%%%%%%%%%

\section{Conclusions}

%This paper provided a systematic passivity analysis of the Epileptor neural mass model of brain dynamics under epilepsy. Using a combination of linear and nonlinear analyses, we proved that seizure dynamics can be both stabilized and passivated via linear passive feedback, despite the open-loop model lacking the internal dissipation as well as input-output matching conditions necessary for passivity. As such, this work provides the first rigorous analysis of the use of PBC for seizure suppression and lays a solid foundation for its further development as a novel treatment approach for drug-resistant epilepsy.Future research is needed to integrate PBC with predictive seizure forecasting and subject-specific, data-driven system identification, as well as the experimental testing of the effect of PBC in animal models and, eventually, patients with epilepsy.

This paper provided a systematic passivity analysis of the Epileptor neural mass model of brain dynamics under epilepsy. Using a combination of linear and nonlinear analyses, we showed that seizure dynamics can be stabilized and rendered passive via linear feedback, despite the open-loop model lacking the internal dissipation and input-output matching conditions required for passivity. These results provide the first rigorous foundation for PBC in seizure suppression and establish fundamental structural limitations of input-output matching in this setting. Future work will focus on integrating PBC with predictive seizure forecasting and subject-specific, data-driven system identification, as well as experimental validation in animal models and, ultimately, patients with epilepsy.
%%%%%%%%%%%%%%%%%%%%%%%%%%%%%%%%%%%%%%%%%%%%%%%%%%%%%%%%%%%%%%%%%%%%%%%%%%%%%%%%
\section{Acknowledgements}

This work was supported in part by National Science Foundation Award No. 2239654.

% \section*{APPENDIX A}
% Details of parameter values are provided here.

% \section*{APPENDIX B}

% \begin{figure}[!ht]
%     \centering
%     \includegraphics[width=\linewidth]{figures/xstar_vs_vref.png}
%     \caption{Stable nodes under PSF vs $y_p$}
%     \label{fig:placeholder}
% \end{figure}

%%%%%%%%%%%%%%%%%%%%%%%%%%%%%%%%%%%%%%%%%%%%%%%%%%%%%%%%%%%%%%%%%%%%%%%%%%%%%%%%

\bibliographystyle{ieeetr}
\bibliography{local}

\end{document}